\documentclass{amsart}

\usepackage{amsfonts,graphicx, amssymb, amsthm, mathtools, hyperref,caption,subcaption} 
\usepackage[section]{placeins}

\newtheorem{theorem}{Theorem}

\newtheorem{lemma}[theorem]{Lemma}

\theoremstyle{definition}

\theoremstyle{remark}

\theoremstyle{remark}
\newtheorem{example}{Example}
\numberwithin{equation}{section}

\newcommand\mult{\operatorname{mult}}

\title[Localization of \v{S}eba billiard eigenfunctions ]{Eigenvalues of \v{S}eba billiards with localization of low-energy eigenfunctions}
\author{Minjae Lee}
\address{Department of Mathematics, University of California, Berkeley, CA 94720, USA}
\email{lee.minjae@math.berkeley.edu}
\allowdisplaybreaks

\begin{document}

\maketitle

\begin{abstract}
We study the localization of eigenfunctions produced by a point scatterer on a thin rectangle. We find an explicit set of eigenfunctions localized to part of the rectangle by showing that the one-dimensional Schr\"{o}dinger operator with a Dirac delta potential asymptotically governs the spectral properties of the two-dimensional point scatterer. We also find the rate of localization in terms of the aspect ratio of the rectangle. In addition, we present numerical results regarding the asymptotic behavior of the localization.
\end{abstract}



\section{Introduction}

The \v{S}eba billiard, a point scatterer on a two-dimensional plate, was first introduced in \cite{seba} as a simple model to study quantum chaos. It is formally given by the Schr\"{o}dinger operator 
\begin{equation}\label{2ddelta}-\Delta +c\delta_{\mathbf{x}_0}\end{equation} where \(\Delta\) is the Dirichlet Laplacian, \(\mathbf{x}_0\) is a point in the plate and \(c\in\mathbb{R}\) but requires renormalization or spectral theory to be properly defined. More precisely, a point scatterer at \(\mathbf{x}_0\) is defined as a self-adjoint extension \cite{reedsimon2} of the Laplacian whose domain consists of functions vanishing near \(\mathbf{x}_0\). Such an extension is not unique and we can parametrize all possible self-adjoint extensions with a single parameter \(\alpha \in (-\infty,\infty]\) which can be considered as a number related to the strength of the point scatterer. See \cite{Shigehara1,Shigehara2,rudnick, Ueberschar} for further developments.

In this paper we consider a point scatterer on a \textit{thin} rectangle. The spectral theory of the Laplacian on a thin domain has been studied in mathematical physics as a model for waveguides, nanotubes, quantum wires and integrated circuits. Such objects are called quasi one-dimensional since its length in one direction is significantly larger than those in the other directions. The simplest example of such domain is a rectangle in which the \textit{eccentricity},  the ratio of the width to the height, is very large. There are other examples such as a cylindrical shell, \(\epsilon\)-neighborhood of a curve and \(\epsilon\)-neighborhood of a graph. Then it is natural to ask if a 1D model on a curve or a graph corresponding to the thin domain can approximate the spectral properties of the 2D operator such as the Laplacian. This question for the Laplacian has been answered in various contexts. See \cite{thin} for details. See also \cite{review1, irregular, fractal, classicalwaves, curvedstrip} for the spectral properties of the Laplacian on other kinds of domains.

In \cite{minjae}, we showed numerically that a point scatterer located inside a thin rectangle acts as a barrier so that the low-energy eigenfunctions except the lowest one get localized as the eccentricity tends to infinity. This work was motivated by the numerical study of Filoche and Mayboroda \cite{bilaplacian} for bi-Laplacian \(\Delta^2\) on a rectangle with a point removed. This fourth order operator exhibits localized modes under the boundary conditions requiring the mode and its gradient to vanish at the boundary. 

In this paper, we provide an explicit set of eigenvalues and corresponding strength parameters of the point scatterer exhibiting the localization of eigenfunctions with the sharp estimate of the error term.  Interestingly, we observed that for low-energy eigenvalues and eigenfunctions there exists asymptotic correspondence between the \v{S}eba billiard and the one-dimensional Schr\"{o}dinger operator with a Dirac delta potential. Furthermore, the same result holds when we impose other boundary conditions such as Neumann, periodic, and Floquet boundary condition on the longer sides of the rectangle. In Section~\ref{sec:1d}, we review a one-dimensional model which turns out to be closely related to low-energy eigenfunctions of \v{S}eba billiards. In Section~\ref{sec:seba}, we consider the optimal estimate measuring the ratio of the localization induced by a point scatterer with an arbitrary parameter. Then in Section~\ref{sec:mainthm}, we present the main theorem by combining the results of two previous sections. In Section~\ref{sec:num}, we will give numerical results regarding the asymptotic behavior of the localization mentioned in the main theorem for the \v{S}eba billiard. Rigorous proofs of the main theorem and the intermediate steps are attached in Section~\ref{proofs}.

\section{One-dimensional model with a Dirac delta potential}\label{sec:1d}
Let \(0<x_0<a,~c\in\mathbb{R}\) and consider a Schr\"{o}dinger operator \(-\Delta-c\delta_{x_0}\) on \([0,a]\) with the Dirichlet boundary condition. We assume that \(\frac{x_0}{a}\) is irrational so that \(\sigma(-\Delta)\), the spectrum of the Dirichlet Laplacian, and \(\sigma(-\Delta-c\delta_{x_0})\) are disjoint. If \(\frac{x_0}{a}\) is rational, then there exist Laplacian eigenfunctions \(\sin\left(\frac{nx}{a}\right)\) vanishing at \(x_0\) for some integer \( n\ge 1\). Since these eigenfunctions do not feel the presence of the delta impurity, they remain as the eigenfunctions of \(-\Delta-c\delta_{x_0}\) with the same eigenvalues. If \(\frac{x_0}{a}\) is irrational, then \(\sin\left(\frac{nx}{a}\right)\) does not vanish at \(x_0\) for any integer \(n \ge 1\) so it cannot be the eigenfunction of \(-\Delta-c\delta_{\mathbf{x}_0}\). The eigenvalue has to change accordingly as well. This implies that \[\sigma(-\Delta) \cap \sigma(-\Delta-c\delta_{x_0}) \ne \emptyset ~\text{ if and only if }~ \frac{x_0}{a} \in \mathbb{Q}.\] In other words, \(\frac{x_0}{a}\) should be irrational in order to keep \(\sigma(-\Delta) \cap \sigma(-\Delta-c\delta_{x_0})=\emptyset\).

The eigenvalues \(z\in\sigma(-\Delta-c\delta_{x_0})\) are given as the solutions to
\begin{equation}\label{1dev}
\sqrt{z}\left(\cot \left(\sqrt{z}x_0\right) + \cot \left(\sqrt{z}(a-x_0)\right)\right)=c
\end{equation} having multiplicity 1 and the corresponding eigenfunctions \(\psi_z^{(1)}\in L^2([0,a])\)
\begin{equation}\label{psi1}
\psi_z^{(1)}(x)= \begin{cases}
N\left[ \cot (\sqrt{z}x_0)\sin (\sqrt{z}(x-x_0)) +\cos (\sqrt{z}(x-x_0))\right], &0<x<x_0\\ N\left[-\cot (\sqrt{z}(a-x_0))\sin (\sqrt{z}(x-x_0)) +\cos (\sqrt{z}(x-x_0))\right], &x_0<x<a\\
\end{cases}\end{equation}
where \(N>0\) is the normalization constant so that \(\|\psi_z^{(1)}\|_{L^2([0,a])}=1\). Note that the superscript \((1)\) of \(\psi_z^{(1)}\) indicates that the eigenfunction is from the one-dimensional model. For each \(c \in \mathbb{R}\), we denote the eigenvalues of \(-\Delta-c\delta_{x_0}\) by \(z_{1,c}<z_{2,c}<z_{3,c}<\cdots\). 

Now we consider the limits of eigenvalues and eigenfunction as \(c\rightarrow \infty\). By \eqref{1dev}, we obtain an increasing sequence \(z_{n,\infty}=\lim_{c\rightarrow\infty} z_{n,c},~n\ge 1\) defined as a union of two distinct sets, namely,
\begin{equation}\label{zinf}
\left\{z_{1,\infty}, z_{2,\infty}, z_{3,\infty}, \cdots \right\}= S_1 \cup S_2 
\end{equation}
where 
\begin{equation} \label{s1s2}
S_1 =\left\{ \left(\frac{m\pi}{x_0}\right)^2 ~\middle|~ m \ge 1\right\} ,\quad S_2=\left\{\left(\frac{m\pi}{a-x_0}\right)^2~ \middle| ~m\ge 1 \right\}.\end{equation}

\emph{Remark.}
For a fixed \(c\in\mathbb{R}\), \eqref{1dev} provides an increasing sequence of eigenvalues \(z_{n,c}>0,~n=1,2,\cdots.\). On the other hand, for any \(z\in (0,\infty)\setminus\{z_{n,\infty}~|~n\ge 1\}\), there exists a unique \(c\in\mathbb{R}\) also given by \eqref{1dev} such that \(z\in\sigma(-\Delta-c\delta_{x_0})\). See Fig.~\ref{1dgraph}.
In addition, since \(0<\frac{x_0}{a}<1\) is irrational, \(\left\{z_{n,\infty} ~|~ n\ge 1 \right\}\) interlaces with a sequence \(\left\{\left(\frac{n \pi}{a}\right)^2 ~\middle|~ n\ge 0 \right\}\), i.e.,\begin{equation}\label{1dinterlace}0<z_{1,\infty}<\frac{\pi^2}{a^2}<z_{2,\infty}<\frac{4\pi^2}{a^2}<z_{3,\infty}<\frac{9\pi^2}{a^2}<\cdots \end{equation}

\begin{figure} 
\centering
\includegraphics[width=0.8\columnwidth]{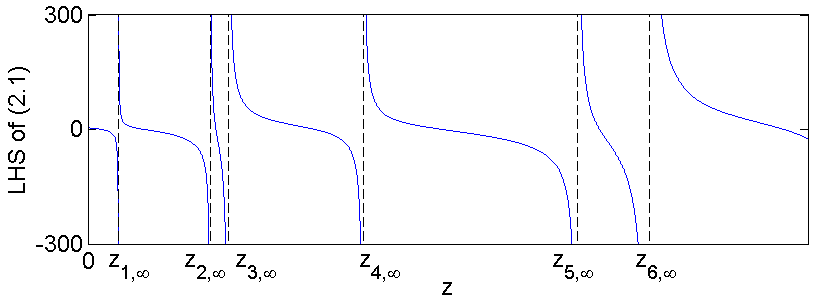}
\caption{A schematic graph of LHS of \eqref{1dev} as a function of \(z\). \(z_{n,\infty}\) in \eqref{zinf} are marked as dashed vertical lines.}
\label{1dgraph}
\end{figure}

As \(c\rightarrow \infty\), the eigenfunctions tend to localize either on \([0,x_0]\) or on \([x_0,a]\) in the \(L^2\)-sense. We now allow \(c\) to take the value \(+\infty\) and let \(\psi_{z_{n,\infty}}^{(1)}(x)=\lim_{c\rightarrow\infty}\psi_{z_{n,c}}^{(1)}(x)\). More precisely, if \(z_{n,\infty} \in S_1\), then
\[\psi_{z_{n,\infty}}^{(1)}(x)=\lim_{c\rightarrow\infty}\psi_{z_{n,c}}^{(1)}(x)=\begin{cases}
\sqrt{\frac{2}{x_0}} \sin \left(\sqrt{z_{n,\infty}}(x-x_0)\right), &0<x<x_0 \\
0, & x_0<x<a
\end{cases}\]
If \(z_{n,\infty} \in S_2\), then
\[\psi_{z_{n,\infty}}^{(1)}(x)=\lim_{c\rightarrow\infty}\psi_{z_{n,c}}^{(1)}(x)=\begin{cases}
0, &0<x<x_0 \\
\sqrt{\frac{2}{a-x_0}} \sin \left(\sqrt{z_{n,\infty}}(x-x_0)\right), & x_0<x<a
\end{cases}\]

Hence,
\begin{equation}\label{1dloc}
\frac{\|\psi_{z_{n,\infty}}^{(1)}\|_{L^2([0,x_0])}}{\|\psi_{z_{n,\infty}}^{(1)}\|_{L^2([0,a])}} = \begin{cases}
1, \mbox{ if } z_{n,\infty} \in S_1 \\ 
0, \mbox{ if } z_{n,\infty} \in S_2
\end{cases}
\end{equation}

On the other hand, for \(z>0\), \(\psi_z^{(1)}\) has the \(L^2\)-expansion:
\begin{equation}\label{psiz1fourier}
\psi_z^{(1)}(x) = M\sum_{n=1}^\infty c_n \sin \left( \frac{n\pi}{a}x \right),\quad c_n = \frac{\sin\left(\frac{n\pi}{a}x_0\right)}{\left(\frac{n\pi}{a}\right)^2-z}
\end{equation}
where \(M\) is the \(L^2\)-normalization constant. Note that if we assume that \(\frac{x_0}{a}\) and \(za^2\) are constant, then 
\begin{equation}\label{psi1norm}
M \propto a^{-5/2}
\end{equation}
since
\[1=\|\psi_z^{(1)}\|_{L^2[0,a]}^2=M^2 \sum_{n=1}^\infty \left(\frac{\sin\left(n\pi \dfrac{x_0}{a}\right)}{\dfrac{n\pi}{a^2}-\dfrac{c}{a^2}}\right)^2 \frac{a}{2}=C M^2a^5\]
for some constant \(C\). We will use \eqref{psi1norm} to prove Lemma~\ref{lemma}  in Section~\ref{proofs}.

\section{A point scatterer on a rectangle} \label{sec:seba}
\subsection{Boundary conditions}\label{boundary}
We consider a point scatterer with various boundary conditions for which the presence of a point scatterer induces the same type of localization of eigenfunctions as in \cite{minjae} so that the \v{S}eba billiard becomes one example of this generalization. 

More precisely, we continue imposing the Dirichlet boundary conditions on \([0,a]\times\{0,b\} \subset\partial\Omega\) and assume an arbitrary boundary condition on \(\{0,a\}\times[0,b]\subset\partial\Omega\) satisfying the following properties: For \(-\Delta\) on \(L^2(\Omega)\) with the boundary condition, the eigenvalues \(\lambda_1\le\lambda_2\le\lambda_3\le\cdots\) and the eigenfunctions \(\phi_1, \phi_2, \phi_3,\cdots \in L^2(\Omega)\) are equal to
\begin{equation}\label{lambdan}
\lambda_n=\lambda_{n_1,n_2}= \left(\frac{n_1 \pi}{a}\right)^2+ \frac{\nu_{n_2}}{b^2},\quad n\ge 1, ~n_1\ge 1, ~n_2\ge 1
\end{equation}
and
\begin{multline}\label{phin}
\phi_n(x,y) = \phi_{n_1,n_2}(x,y)=\sin\left(\frac{n_1 \pi}{a}x\right) g_{n_2}\left(\frac{y}{b}\right), \\ 0\le x \le a, ~ 0\le y \le b,~ n\ge 1, ~n_1\ge 1, ~n_2\ge 1
\end{multline}
for some \(0\le \nu_1 \le \nu_2  \le  \cdots \) and a set of nonzero orthogonal functions \(\{g_n~|~ n \ge 1\}\subset L^2([0,1])\) such that 
\begin{equation}\label{weyl}
\#\{\nu_n~|~\nu_n<\nu, ~n\ge 1\}=O(\sqrt{\nu}),~\nu\rightarrow\infty.
\end{equation}

For example, the Dirichlet, Neumann, periodic and the Floquet boundary conditions on \([0,a]\times\{0,b\}\subset\partial\Omega\) all satisfy these conditions. More precisely,
\[
\Delta u +\lambda u =0 \quad \text{ in } \Omega, \quad \lambda\in\mathbb{R},\quad u(0,y)=u(a,y)=0, ~ 0\le y \le b
\]
has the eigenvalues \(\lambda_n\) and the eigenfunctions \(\phi_n\) determined by \(\nu_n\) and \(g_n\) as in  \eqref{lambdan}, \eqref{phin} for each boundary condition on \([0,a]\times\{0,b\}\subset\partial\Omega\). 

\begin{enumerate}
\item Dirichlet: \( u(x,0)=u(x,b)=0, \quad 0\le x \le a \)
\[\nu_n=\left(n\pi\right)^2, ~g_n(y)=\sin(n\pi y) , ~n=1,2,3,\cdots \]
\item Neumann: \(\partial_y u(x,0)=\partial_y  u(x,b)=0, \quad 0\le x \le a \)
\[\tilde\nu_n=\left(n\pi\right)^2, ~\tilde g_n(y)=\begin{cases}\frac{1}{2}, ~&n=0\\ \cos(n\pi y),~ &n=1,2,\cdots \end{cases}\]
\item Periodic: \( u(x,0)=u(x,b), \quad 0\le x \le a\)
\[\tilde\nu_n=\left(2n\pi\right)^2, ~\tilde g_n(y)=e^{2i n\pi y} , ~ n\in\mathbb{Z}\]
\item Floquet: \( u(x,0)=e^{-i\theta}u(x,b), \quad 0\le x \le a\) for some \(\theta\in (-\pi,\pi)\)
\[ \tilde\nu_n=\left(2n\pi+\theta\right)^2, ~\tilde g_n(y)=e^{i (2n\pi+\theta) y} , ~ n\in\mathbb{Z}\]
\end{enumerate}
For (2),(3) and (4), we define \(\nu_n, g_n\) \( (n\ge 1)\) by rearranging \(\tilde \nu_n\) (\(n\ge 0 \) or \(n\in\mathbb{Z}\)) in the nondecreasing order with a new index \(n \ge 1 \) so that \eqref{weyl} holds. 


\FloatBarrier

\subsection{Spectral properties of a point scatterer}
\begin{figure}
\centering
\includegraphics[width=0.5\columnwidth]{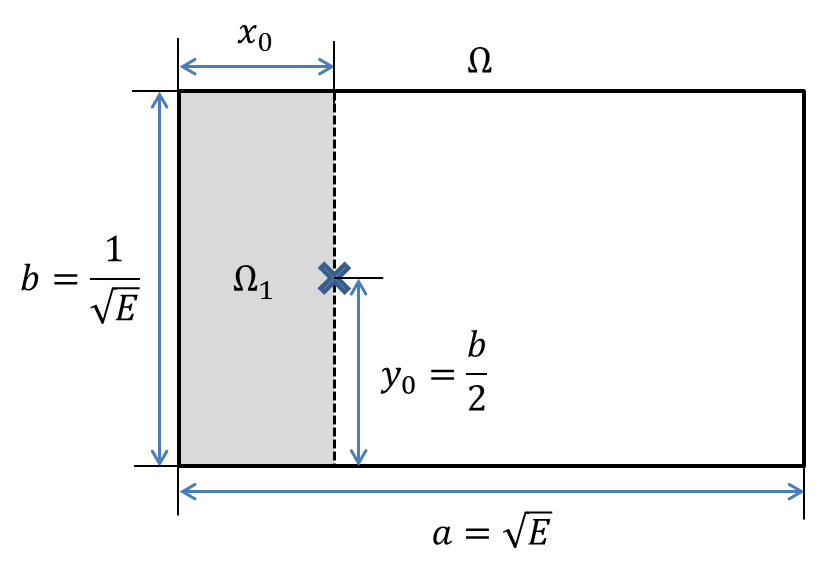}
\caption{Geometry of a point scatterer at \(\mathbf{x}_0=(x_0,y_0) \) (marked as \(\times\)) in \(\Omega\). The left part of the plate divided by \(\mathbf{x}_0\) is denoted by \(\Omega_1=[0,x_0]\times[0,b]\). }
\label{plate}
\end{figure}

 Now we construct a point scatterer at \(\mathbf{x}_0=(x_0,\frac{b}{2})\) on a rectangular plate \(\Omega=[0,a]\times[0,b]\) as in Fig.~\ref{plate} with the boundary conditions discussed in Section~\ref{boundary}. We continue assuming \(\frac{x_0}{a}\) to be irrational as in Section~\ref{sec:1d}.

 First, consider the Laplacian \(-\Delta\) on \(L^2(\Omega)\) with the boundary condition satisfying \eqref{lambdan}, \eqref{phin}, \eqref{weyl}. Then restrict its domain to the functions vanishing at \(\mathbf{x}_0\). By the theory of self-adjoint extension developed by Von Neumann, such a symmetric operator has a family of self-adjoint extensions \(-\Delta_{\mathbf{x}_0,\Omega,\alpha}\) with a parameter \(\alpha\in (-\infty,\infty]\). The parameter \(\alpha\) is following Albeverio's notation from \cite{solvable} as in our previous paper \cite{minjae}. Note that \(\alpha\) determines the strength of the point scatterer at \(\mathbf{x}_0\in\Omega\) although it is not equal to the coefficient \(c\in\mathbb{R}\) of the delta potential in \eqref{2ddelta}. More precisely, \(-\Delta+c\delta_{\mathbf{x}_0}\) needs renormalization to be properly defined as a self-adjoint operator on \(L^2(\Omega)\). Let \(G_z\) be the integral kernel of the resolvent \((-\Delta-z)^{-1}: L^2(\Omega)\rightarrow L^2(\Omega) \), namely,
\[G_z(\mathbf{x},\mathbf{x}')=\sum_{n=1}^\infty \frac{\phi_n(\mathbf{x})\overline{\phi_n(\mathbf{x}')}}{\lambda_n-z}\]
so that for \(f\in L^2(\Omega)\),
\[(-\Delta-z)^{-1}f(\mathbf{x}) = \int_\Omega G_z(\mathbf{x},\mathbf{x}') f(\mathbf{x}') d\mathbf{x}'.\]

Then the integral kernel of \(-\Delta +c\delta_{\mathbf{x}_0}, ~c\in\mathbb{R}\) formally reads
\[(-\Delta +c\delta_{\mathbf{x}_0}-z)^{-1}(\mathbf{x},\mathbf{x}')=G_z(\mathbf{x},\mathbf{x}') -\left[\frac{1}{c}-G_z(\mathbf{x}_0,\mathbf{x}_0)\right]^{-1} G_z(\mathbf{x}_0,\mathbf{x}')G_z(\mathbf{x},\mathbf{x}_0).\]
However, this approach fails if \(G_z(\mathbf{x}_0,\mathbf{x}_0)\) diverges which did not occur in the one-dimensional case. Therefore, we have to discard \(c\) and renormalize the delta potential with a new parameter \(\alpha\in (-\infty,\infty]\) to obtain the resolvent formula for the point scatterer \(-\Delta_{\mathbf{x}_0,\Omega,\alpha}\). For \(z \in \mathbb{C}\setminus \sigma (-\Delta_{\mathbf{x}_0,\Omega,\alpha})\), the integral kernel of \[(-\Delta_{\mathbf{x}_0,\Omega,\alpha}-z)^{-1}:L^2(\Omega)\rightarrow L^2(\Omega)\] is given by
\begin{multline}\label{resolvent}
(-\Delta_{\mathbf{x}_0,\Omega,\alpha}-z)^{-1}(\mathbf{x},\mathbf{x}') \\= G_z(\mathbf{x},\mathbf{x}')+ [\alpha - F(z)]^{-1} G_z(\mathbf{x}_0,\mathbf{x}')G_z(\mathbf{x},\mathbf{x}_0), \quad \alpha\in (-\infty,\infty]
\end{multline}
where
\begin{equation}\label{Fz}
F(z)= \sum_{n=1}^{\infty} |\phi_n(\mathbf{x}_0)|^2\left(\frac{1}{\lambda_n -z}-\frac{\lambda_n}{\lambda_n^2 +1}\right).
\end{equation}
is derived by applying Von Neumann's self-adjoint extension theory \cite{reedsimon2} to the Laplacian restricted to the smooth functions vanishing at \(\mathbf{x}_0\).

According to this parametrization, the point scatterer disappears in the norm resolvent sense of becoming \(-\Delta\) as \(\alpha \rightarrow \pm\infty\)  whereas it gets stronger when \(|\alpha|\ll \infty\). In particular, \(\alpha=\infty\) corresponds to the  Laplacian with the same boundary condition on \(\partial\Omega\) so we will only consider \(\alpha \in \mathbb{R}\) for which the presence of the point scatterer actually perturbs the system. We may also interpret \(\alpha\) as the inverse of the coupling constant \(v_B\) or \(\overline{v}_\theta\)  in Shigehara's work \cite{Shigehara1,Shigehara2}. For more details on the construction of \eqref{resolvent}, \eqref{Fz}, see \cite{solvable}.

\begin{figure} 
\centering
\includegraphics[width=0.6\columnwidth]{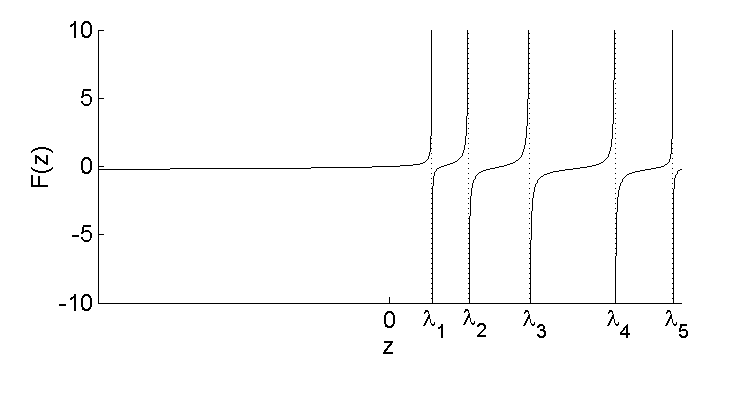}
\caption{A schematic graph of \(F(z)\) of \eqref{Fz}. \(\lambda_n\) in \eqref{lambdan} are marked as dotted vertical lines.}
\label{Fz2}
\end{figure}

Now consider the eigenvalues of the point scatterer with the boundary conditions in Section~\ref{boundary}. For simplicity let us omit \(\mathbf{x}_0\) and \(\Omega\) from \(-\Delta_{\mathbf{x}_0,\Omega,\alpha}\), namely,
\[-\Delta_\alpha=-\Delta_{\mathbf{x}_0,\Omega,\alpha}\] since those quantities are fixed. Also, we denote \(\mult(z,P)\) by the multiplicity of \(z\) as an eigenvalue of an operator \(P\). For \(\alpha\in\mathbb{R}\), we can divide \(\sigma(-\Delta_\alpha)\) into two types:
\begin{enumerate}
\item Perturbed eigenvalues: \(\sigma(-\Delta_{\alpha})\setminus\sigma(-\Delta) \) and
\item Unperturbed eigenvalues: \(\sigma(-\Delta_{\alpha})\cap\sigma(-\Delta)\)
\end{enumerate}
where each of them is obtained by different conditions. 
For \(\alpha\in\mathbb{R}\), \begin{equation}\label{alphaFz}
z\in\sigma(-\Delta_{\alpha})\setminus \sigma(-\Delta) \quad \text{ if and only if }\quad
\alpha=F(z)
\end{equation}
with \(F(z)\) defined in \eqref{Fz}. See Fig.~\ref{Fz2}.
Furthermore, \(\mult(z,-\Delta_\alpha)=1 \).

On the other hand, define \(\mu\) and \(\mu_0\) as 
\begin{align*}\mu(z) &\equiv \mult(z,-\Delta) = \#\{n\ge 1 ~|~ z=\lambda_n \} \\
\mu_0(z) &\equiv \#\{n\ge 1 ~|~ z=\lambda_n, \phi_n(\mathbf{x}_0)=0 \}. \end{align*}
Then for \(\alpha\in\mathbb{R}\), 
\begin{equation}\label{unpert}
z\in\sigma(-\Delta_\alpha) \cap \sigma(-\Delta) \quad \text{if and only if}\quad \mu_0(z)\ge 1 \text{ or } \mu(z)\ge 2.
\end{equation}
with the corresponding eigenspace 
\begin{equation}\label{unperteigftn}
\left\{ \sum_{z=\lambda_n} c_n \phi_n ~\middle|~ \sum_{z=\lambda_n} c_n \phi_n(\mathbf{x}_0)=0,~ c_n\in\mathbb{C} \right\}.
\end{equation}
By \eqref{unpert} and \eqref{unperteigftn}, we observe that the unperturbed eigenvalues and the corresponding eigenfunctions both are independent of \(\alpha\in\mathbb{R}\). The proofs of \eqref{unpert} and \eqref{unperteigftn} can be found in Chapter 2 of \cite{pseudolaplacian} with generalized statements for a point scatterer on a compact Riemannian manifold of dimension two or three.

 Let \(z_{1,\alpha}' \le  z_{2,\alpha}' \le z_{3,\alpha}' \le \cdots\) be the eigenvalues of \(-\Delta_\alpha\). Now we remark some properties of \(\sigma(-\Delta_{\alpha})\).

\emph{Remark.}
For any \(z\in \mathbb{R}\setminus\{\lambda_{n}~|~n\ge 1\}\), there exists a unique \(\alpha\in\mathbb{R}\) given by \eqref{alphaFz} such that \(z\in\sigma(-\Delta_\alpha)\).
For any \(\alpha\in \mathbb{R}\), \(\sigma(-\Delta_\alpha)\) interlaces with \(\sigma(-\Delta)\), namely,
\begin{equation}\label{interlace}z'_{1,\alpha} \le \lambda_1 \le z'_{2,\alpha} \le \lambda_2 \le \cdots .\end{equation}
In addition, for \(n \ge 1\), we have
\[\lim_{\alpha\rightarrow\infty}z'_{n,\alpha} = \lambda_n ,\qquad 
\lim_{\alpha\rightarrow -\infty}z'_{n+1,\alpha} = \lambda_n, \qquad 
\lim_{\alpha\rightarrow -\infty}z'_{1,\alpha} =-\infty.\]

 Let \(\psi_z^{(2)}\) be the eigenfunction associated with  \(z\in\sigma(-\Delta_\alpha)\setminus\sigma(-\Delta)\). Note that the superscript \((2)\) of \(\psi_z^{(2)}\) indicates that the eigenfunction is for the two-dimensional model.
Then we have \[\psi_z^{(2)}(\mathbf{x})= M G_z (x_0,y_0;~x,y)\] where 
\(G_z\) is the integral kernel of the resolvent operator \( (-\Delta-z)^{-1} : L^2 (\Omega)\rightarrow L^2(\Omega)\) and \(M>0\) is the \(L^2\)-normalization constant.
In addition, \(\psi_z^{(2)}\) has the \(L^2\)-expansion up to a multiplicative constant:
\begin{equation}\label{psiz2}
\psi_z^{(2)}= \sum_{n=1}^\infty \frac{\overline{\phi_n(x_0,y_0)} }{\lambda_n-z}\phi_n = \sum_{n_1=1}^\infty \sum_{n_2=1}^\infty \frac{\overline{\phi_{n_1,n_2}(x_0,y_0) }}{\lambda_{n_1,n_2}-z}\phi_{n_1,n_2} 
\end{equation}
Note that for \(\alpha\in \mathbb{R}\), any \(f\) in the domain of \(-\Delta_\alpha\)  has a logarithmic singularity  at \((x_0,y_0)\) or vanish at \((x_0,y_0)\). More precisely, as \((x,y)\rightarrow (x_0,y_0)\), 
\begin{equation}\label{logdiv}
f(x,y)= C_1 \left(\ln\sqrt{(x-x_0)^2+(y-y_0)^2} +C_2(\alpha) \right)+o(1)
\end{equation} for some \(C_1\in\mathbb{R}\) and \(C_2(\alpha)\in\mathbb{R}\) fixed for each \(\alpha\).
 Proofs can be found in Chapter 1 of \cite{pseudolaplacian} with a general statement for a complete and smooth Riemannian manifold of dimension two. This implies that \(\psi_z^{(2)}\) also contains the logarithmic singularity unless \(\psi_z^{(2)}(x,y)=o(1)\) as \((x,y)\rightarrow (x_0,y_0)\).

\emph{Remark.}
Without loss of generality, we may assume \(\mathrm{area}(\Omega)=ab=1\) or equivalently, \[\quad a=\sqrt{E}, \quad b=\frac{1}{\sqrt{E}}\] by scaling \(\Omega\), \(z_{n,\alpha}'\) and \(\psi_{z_{n,\alpha}'}^{(2)}\) simultaneously according to the following properties:
Consider point scatterers \(-\Delta_{\mathbf{x}_0,\Omega,\alpha}\) and \(-\Delta_{r\mathbf{x}_0,r\Omega, \alpha}\) on \(\Omega=[0,a]\times[0,b]\) and \(r\Omega=[0,ra]\times[0,rb]\), respectively, with a coupling constant \(\alpha\). Then 
\begin{equation}\label{scaling}z\in\sigma(-\Delta_{r\mathbf{x}_0,r\Omega,\alpha}) ~\text{ if and only if }~ r^2 z\in\sigma(-\Delta_{\mathbf{x}_0,\Omega,\alpha-\beta})\end{equation}
where \(\beta\) is a constant defined as
\[\beta=\sum_{n=1}^\infty \left(\frac{\lambda_n^2}{\lambda_n^2+1}-\frac{\lambda_n^2}{\lambda_n^2+r^4}\right).\]
In addition, let \(\tilde{\psi}_z^{(2)}\) and \(\psi_{r^2z}^{(2)}\) be the eigenfunctions corresponding to the eigenvalues \(z \in-\Delta_{r\Omega,\alpha}\) and \(r^2 z \in -\Delta_{\Omega, \alpha-\beta}\), respectively. Then for some \(C\ne 0\), we have \[\tilde{\psi}_z^{(2)}(x,y) = C \psi_{r^2z}^{(2)}\left(\frac{x}{r},\frac{y}{r}\right).\]

Now we estimate the localization of eigenfunctions over \(\Omega_1 = [0,x_0]\times [0,b]\) in the \(L^2\)-sense where the perturbed eigenvalue is bounded below by the lowest eigenvalue of the Laplacian and bounded above by \(\min \{\lambda_{1,n}~|~\lambda_{1,n}>\lambda_{1,1},~n\ge 1\}\).  See Section~\ref{proofs} for the proofs.

\begin{lemma}\label{lemma}
Define \(\lambda_{n_1,n_2},  \nu_{n_2}\) as in \eqref{lambdan} and let \(a=\sqrt{E},~ b=\frac{1}{\sqrt{E}}\). Let \begin{equation}\label{ntilde}
\tilde\lambda=\min \{\lambda_{1,n}~|~\lambda_{1,n}>\lambda_{1,1},~n\ge 1 \},\qquad \tilde{n}=\min\{n\ge 1~|~\lambda_{1,n}>\lambda_{1,1}\}.\end{equation}
 Let \(z\in (\lambda_{1,1},\tilde\lambda) \cap \sigma(-\Delta_\alpha) \setminus \sigma(-\Delta)\) and assume \(z=z(E)=\frac{c}{E}+\nu_1E\) for some constant \(c\in\mathbb{R}\). Then



\begin{equation}\label{eqn:lemma1}\frac{\|\psi_z^{(2)}\|_{L^2(\Omega_1)}}{\|\psi_z^{(2)}\|_{L^2(\Omega)}} = 
\frac{\|\psi_{z-\nu_1 E} ^{(1)}\|_{L^2([0,x_0])}}{\|\psi_{z-\nu_1 E} ^{(1)}\|_{L^2([0,a])}} + O\left(\frac{1}{E \sqrt{\frac{\pi^2}{E}+\tilde{n}^2 C_0 E -z}}\right)\end{equation}
and 
\begin{equation}\label{eqn:lemma2}\frac{\|\psi_z^{(2)}\|_{L^2(\Omega\setminus\Omega_1)}}{\|\psi_z^{(2)}\|_{L^2(\Omega)}} = 
\frac{\|\psi_{z-\nu_1 E} ^{(1)}\|_{L^2([x_0,a])}}{\|\psi_{z-\nu_1 E} ^{(1)}\|_{L^2([0,a])}} + O\left(\frac{1}{E \sqrt{\frac{\pi^2}{E}+\tilde{n}^2 C_0 E -z}}\right)\end{equation}
as \(E\rightarrow\infty\), where 
\begin{equation}\label{C0}
C_0=\sup_{n\ge1}\frac{\nu_n}{n^2},\qquad \end{equation}
with \(\psi_z^{(1)}\in L^2([0,\sqrt{E}])\) of \eqref{psi1}. Note that \(z\) and its bounds \(\lambda_{1,1}, \tilde{\lambda}\) all depend on \(E\). 
\end{lemma}


\section{Main Theorem} \label{sec:mainthm}
In this section, we provide a sequence of \(\alpha\)'s and eigenvalues inducing the localization of certain eigenfunctions by combining \eqref{1dloc} with Lemma~\ref{lemma}. 

\begin{theorem}\label{mainthm} 
Let \[N_E=\left[\sqrt{\frac{\tilde\nu-\nu_1}{\pi^2}E^2+1}\right]\] where \(\tilde\nu=\min\{\nu_n ~|~ \nu_n>\nu_1,~ n\ge 1\}\). Choose \(E\) sufficiently large so that \(N_E \ge \tilde{n}\). Define \(S_1\), \(S_2\) and \(z_{n,\infty}\) as in \eqref{s1s2} with \(a=\sqrt{E}\).
For \(n=\tilde n, \cdots, N_E\), there exists a unique parameter \(\alpha_n\) explicitly defined as
\begin{equation}\label{alphan}\alpha_n =\sum_{n'=1}^\infty |\phi_{n'}(x_0,y_0)|^2 \left(\frac{1}{\lambda_{n'}- \nu_1 E - z_{n-1,\infty}}-\frac{\lambda_{n'}}{\lambda_{n'}^2+1}\right).\end{equation}
such that
\[z_{n,\alpha_n}'=\nu_1 E + z_{n-1,\infty} \in \sigma(-\Delta_{\alpha_n}).\] In addition, as \(E\rightarrow \infty\),
\begin{equation}\label{l2ratio}
\begin{dcases}
\frac{\|\psi_{z_{n,\alpha_n}'}^{(2)}\|_{L^2(\Omega\setminus\Omega_1)}}{\|\psi_{z_{n,\alpha_n}'}^{(2)}\|_{L^2(\Omega)}}=
O\left(\frac{1}{\sqrt{E^3(\tilde{n}^2 C_0-\nu_1)-E \pi^2(n^2-1)}}\right),\quad &\mbox{if } z_{n,\infty}\in S_1 \\ 
\frac{\|\psi_{z_{n,\alpha_n}'}^{(2)}\|_{L^2(\Omega_1)}}{\|\psi_{z_{n,\alpha_n}'}^{(2)}\|_{L^2(\Omega)}}=O\left(\frac{1}{\sqrt{E^3(\tilde{n}^2 C_0-\nu_1)-E \pi^2(n^2-1)}}\right),\quad &\mbox{if } z_{n,\infty}\in S_2 \\ 
\end{dcases}\end{equation}
\end{theorem}

Note that \(\tilde{n}\) in \eqref{ntilde} is equal to \(2\) if \(\nu_1 \ne\nu_2\), which is the generic case. Then \(N_E\ge \tilde{n}\) is chosen as the largest number satisfying \[\lambda_n=\lambda_{n,1}=\left(\frac{n\pi}{a}\right)^2+\frac{\nu_1}{b^2}\] and
\[\phi_n (x,y) = \sin\left(\frac{n\pi x}{a}\right) g_1 (y),\quad 0\le x \le a,~ 0\le y \le b\] for all \( n\le N_E\). Since \(\frac{\nu_1}{b^2}\) and \(g_1\) are fixed with respect to \(n\), we can say that \(N_E\) describes how many low-level modes of the 2D harmonic vibration on \(\Omega\) can be approximated by those of the 1D harmonic vibration on \([0,a]\) with some constant quantities such as \(\frac{\nu_1}{b^2}\) and \(g_1\). By Theorem~\ref{mainthm}, this idea can be extended to approximating a point scatterer on \(\Omega\) by a Schr\"{o}dinger operator with a delta potential on \([0,a]\) in which the localization of eigenfunction has been already exhibited. However, the lowest mode  of the point scatterer does not have a corresponding mode in the 1D model so \(n <\tilde{n}\) (or \(n=1\) in generic cases) should be excluded from the argument. See the proof in Section~\ref{proofs} for more details.  Note that localized eigenfunctions could appear sporadically above \(N_E\) as presented in \cite{minjae}. However, such cases will eventually disappear as \(n\rightarrow\infty\) since the high-energy eigenfunctions of \v{S}eba billiards tend to localize rather in the momentum space. See \cite{local} for the details.

\emph{Remark.}
In the asymptotic notation \eqref{l2ratio}, \(n\) may be a variable depending on \(E\). If \(n\ge 2 \) and \(n=o(E)\) as \(E \rightarrow \infty\), then \eqref{l2ratio} reads as \(E\rightarrow\infty\),
\begin{equation}\label{l2rationfixed}
\begin{dcases}\frac{\|\psi_{z_{n,\alpha_n}'}^{(2)}\|_{L^2(\Omega\setminus\Omega_1)}}{\|\psi_{z_{n,\alpha_n}'}^{(2)}\|_{L^2(\Omega)}}=
O\left(E^{-\frac{3}{2}}\right),\quad &\mbox{if } z_{n,\infty}\in S_1 \\ 
\frac{\|\psi_{z_{n,\alpha_n}'}^{(2)}\|_{L^2(\Omega_1)}}{\|\psi_{z_{n,\alpha_n}'}^{(2)}\|_{L^2(\Omega)}}=O\left(E^{-\frac{3}{2}}\right),\quad &\mbox{if } z_{n,\infty}\in S_2. \\ 
\end{dcases}\end{equation}

In addition, since \[0\le \frac{\|\psi_{z'_{n,\alpha_n}}^{(2)}\|_{L^2(\Omega_1)}}{\|\psi_{z'_{n,\alpha_n}}^{(2)}\|_{L^2(\Omega)}} \le 1,\] \eqref{l2ratio} actually implies that there exists a constant \(C>0\) and \(M>0\) such that for all \(E>M\),
\[\begin{dcases}
0 \le \frac{\|\psi_{z_{n,\alpha_n}'}^{(2)}\|_{L^2(\Omega\setminus\Omega_1)}}{\|\psi_{z_{n,\alpha_n}'}^{(2)}\|_{L^2(\Omega)}} < C\left(\frac{1}{\sqrt{E^3(\tilde{n}^2 C_0-\nu_1)-E \pi^2(n^2-1)}}\right),~ &\mbox{if } z_{n,\infty}\in S_1 \\
0 \le \frac{\|\psi_{z_{n,\alpha_n}'}^{(2)}\|_{L^2(\Omega_1)}}{\|\psi_{z_{n,\alpha_n}'}^{(2)}\|_{L^2(\Omega)}} < C\left(\frac{1}{\sqrt{E^3(\tilde{n}^2 C_0-\nu_1)-E \pi^2(n^2-1)}}\right),~ &\mbox{if } z_{n,\infty}\in S_2. \end{dcases}\]

\begin{figure} 
\centering
\includegraphics[width=0.8\columnwidth]{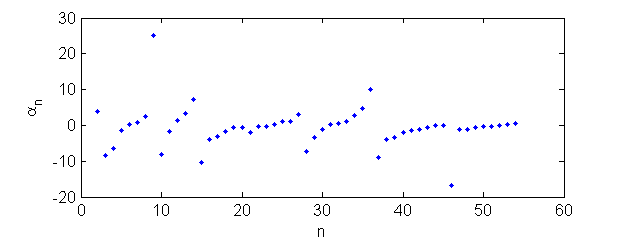}
\caption{Plot of \(\alpha_n\), \(2\le n \le N_E\) given by \eqref{alphan_dirichlet} where \({x_0}=\frac{a}{\pi},~ E=10\pi\) with the Dirichlet boundary condition on \(\partial\Omega\).}
\label{fig:alphan}
\end{figure}

In particular, we may apply Theorem~\ref{mainthm} to the examples mentioned in Section~\ref{boundary} and calculate \(\alpha_n\), \(z'_{n,\alpha_n}\) that generate the localization of modes as exhibited partly in \cite{minjae}. 
 Note that the Dirichlet boundary condition is given on \(\{0,a\}\times[0,b]\) for all cases.

\begin{example}[\v{S}eba billiards]
For the Dirichlet boundary condition on \(\partial\Omega\),
\begin{equation}\label{alphan_dirichlet}\alpha_n =\sum_{n_1=1}^\infty \sum_{\substack{n_2\ge 1\\n_2 \text{ odd}}} \sin^2\left(\frac{n_1 \pi }{\sqrt{E}}x_0\right)\left(\frac{1}{\lambda_{n_1,n_2}- \pi^2 E - z_{n,\infty}}-\frac{\lambda_{n_1,n_2}}{\lambda_{n_1,n_2}^2+1}\right) \end{equation}
and
\[z_{n,\alpha_n}'=\pi^2 E + z_{n-1,\infty} \in \sigma(-\Delta_{\alpha_n}),\quad 2 \le n \le \left[\sqrt{3E^2+1}\right]\]
where
\[\lambda_{n_1,n_2}=\frac{n_1^2\pi^2}{E}+n_2^2\pi^2 E, \quad  n_1\ge 1,~  n_2 \ge 1.\]
See Fig.~\ref{fig:alphan} for the plot of \(\alpha_n\) given by \eqref{alphan_dirichlet} .
\end{example}
\begin{example} 
For the Neumann boundary condition on \([0,a]\times\{0,b\}\subset\partial\Omega\),
\begin{multline}\label{alphan_neumann}\alpha_n =\sum_{n_1=1}^\infty \sin^2\left(\frac{n_1 \pi }{\sqrt{E}}x_0\right) \Biggl[\frac{1}{4} \left(\frac{1}{\tilde\lambda_{n_1,0} - z_{n,\infty}}-\frac{\tilde\lambda_{n_1,0}}{\tilde\lambda_{n_1,0}^2+1}\right)\\+\sum_{\substack{n_2\ge 1\\n_2 \text{ even}}} \left(\frac{1}{\tilde\lambda_{n_1,n_2} - z_{n,\infty}}-\frac{\tilde\lambda_{n_1,n_2}}{\tilde\lambda_{n_1,n_2}^2+1}\right) \Biggr]\end{multline} 
and
\[z_{n,\alpha_n}'= z_{n-1,\infty} \in \sigma(-\Delta_{\alpha_n}),\quad 2 \le n \le \left[\sqrt{E^2+1}\right]\]
where
\[\tilde\lambda_{n_1,n_2}=\frac{n_1^2\pi^2}{E}+n_2^2\pi^2 E, \quad n_1\ge 1,~ n_2 \ge 0.\]
\end{example}

\begin{example} For the periodic boundary condition on \([0,a]\times\{0,b\}\subset\partial\Omega\),
\begin{equation}\label{alphan_periodic}\alpha_n =\sum_{n_1=1}^\infty \sum_{n_2\in \mathbb{Z}}\sin^2\left(\frac{n_1 \pi }{\sqrt{E}}x_0\right) \left(\frac{1}{\tilde\lambda_{n_1,n_2}- z_{n,\infty}}-\frac{\tilde\lambda_{n_1,n_2}}{\tilde\lambda_{n_1,n_2}^2+1}\right)\end{equation}
and
\[z_{n,\alpha_n}'= z_{n-1,\infty} \in \sigma(-\Delta_{\alpha_n}),\quad 2 \le n \le \left[\sqrt{4E^2+1}\right]\]
where
\[\tilde\lambda_{n_1,n_2}=\frac{n_1^2\pi^2}{E}+4n_2^2\pi^2 E, \quad n_1\ge 1,~ n_2 \in\mathbb{Z}.\]
\end{example}

\begin{example} For the Floquet boundary condition on \([0,a]\times\{0,b\}\subset\partial\Omega\) with some \(\theta\in (-\pi,\pi)\),
\begin{equation}\label{alphan_floquet}\alpha_n =\sum_{n_1=1}^\infty \sum_{n_2\in \mathbb{Z}}\sin^2\left(\frac{n_1 \pi }{\sqrt{E}}x_0\right) \left(\frac{1}{\tilde\lambda_{n'}- \theta^2 E - z_{n,\infty}}-\frac{\tilde\lambda_{n'}}{\tilde\lambda_{n'}^2+1}\right)\end{equation}
and
\[z_{n,\alpha_n}'=\theta^2 E + z_{n-1,\infty} \in \sigma(-\Delta_{\alpha_n}),\quad 2 \le n \le \left[\sqrt{4E^2\left(1-\frac{|\theta|}{\pi}\right)+1}\right]\]
where
\[\tilde\lambda_{n_1,n_2}=\frac{n_1^2\pi^2}{E}+(2\pi n_2+\theta)^2 E, \quad n_1\ge 1,~ n_2 \in\mathbb{Z}.\]

\end{example}

\FloatBarrier
\section{Numerical Results for \v{S}eba billiards}\label{sec:num}
In \cite{minjae}, we numerically exhibited several kinds of localized eigenfunctions of the \v{Seba} billiard without specifying \(\alpha_n\), the parameters localizing the \(n\)-th mode up to an error of \(O(E^{-\frac{3}{2}})\) to one side of the rectangle. We now numerically check how accurately the localization occurs for \(-\Delta_{\alpha_n}\) for \(E\) large as expected by Theorem~\ref{mainthm}. For the sake of convenience, let the \(x\)-coordinate of the point scatterer be fixed at \(x_0=\frac{a}{\pi}\). Note that, however, the qualitative properties we observe also hold for other values of \({x_0}\).

First, consider a \v{S}eba billiard with a fixed eccentricity \(E=10\pi\). Although \(\alpha\in\mathbb{R}\) is the variable to be considered for \(-\Delta_\alpha\), we can let the eigenvalue \(z \in \sigma(-\Delta_\alpha)\) itself be an independent variable in \( \mathbb{R}\setminus \{\lambda_n ~|~ n\ge 1\}\) and let \(\alpha\in\mathbb{R}\) depend on \(z\) since for each \(z\), there exists a unique parameter \(\alpha\in\mathbb{R}\) defined by \eqref{alphaFz} such that \(z\in \sigma(-\Delta_\alpha)\). In this point of view, Fig.~\ref{PR_z} shows the \(L^2\)-norm of the eigenfunction \(\psi_z^{(2)}\) over \(\Omega_1\) as a function of \(z\). Note that we assumed \(\|\psi_z^{(2)}\|_{L^2 (\Omega)}=1\) so \(\|\psi_z^{(2)}\|_{L^2 (\Omega_1)}=1 \) and \(\|\psi_z^{(2)}\|_{L^2 (\Omega_1)}=0 \) mean that \(\psi_z^{(2)}\) gets completely localized in \(\Omega_1\) and \(\Omega\setminus \Omega_1\), respectively. The dashed and dotted lines indicate the sets of optimal eigenvalues \(S_1+\pi^2E\) and \(S_2+\pi^2E\) exhibiting the localization in \(\Omega_1\) and \(\Omega\setminus\Omega_1\) expected by Theorem~\ref{mainthm}. One can observe that those lines in Fig.~\ref{PR_z} thoroughly estimate the \(n\)-th lowest eigenvalues (\(n=2,\cdots,N_E\)) of the localized eigenfunctions. For example, Fig.~\ref{mode3} shows the partial sum of the first \(10^6\) terms of \eqref{psiz2} where \(z=z'_{3,\alpha_3}=\pi^2 E + z_{2,\infty}\). Since \(z_{2,\infty}=\left(\frac{\pi}{x_0}\right)^2\in S_1\), the eigenfunction tends to localize on \(\Omega_1\). A small peak at \(\left(x_0,\frac{b}{2}\right)\) indicates the logarithmic divergence of \(\psi_z^{(2)}\) mentioned in \eqref{logdiv}. This phenomenon was not emphasized in the figures of localized modes in \cite{minjae}.

\begin{figure}
\centering
\includegraphics[width=0.8\columnwidth]{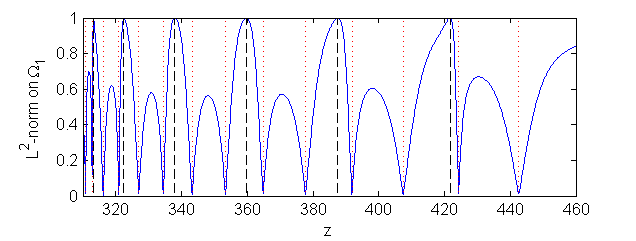}
\caption{(\v{S}eba billiard) \(\|\psi_z^{(2)}\|_{L^2(\Omega_1)}\) as a function of \(z\) where \({x_0}=\frac{a}{\pi},~ E=10\pi\). 
The dashed and dotted vertical lines indicate sets of eigenvalues: \(S_1+\pi^2E\) and \(S_2+\pi^2E\), respectively, in which the localization expected by Theorem~\ref{mainthm} occurs.}
\label{PR_z}
\end{figure}

\begin{figure}
\centering
\includegraphics[width=0.6\columnwidth]{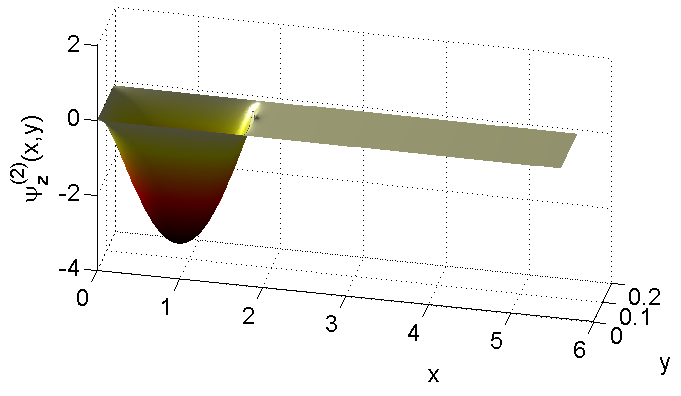}
\caption{(\v{S}eba billiard) Plot of \(\psi_z^{(2)}\) computed as the partial sum of the first \(10^6\) terms of \eqref{psiz2} where \(E=10\pi,~ x_0=\frac{a}{\pi}\approx1.78,~z=z'_{3,\alpha_3}\approx3.13\cdot 10^2\).}
\label{mode3}
\end{figure}

Second, we introduce numerical examples verifying how strongly the eigenfunctions \(\psi_{z'_{n,\alpha_n}}^{(2)}\) localize in terms of the asymptotic estimate in \eqref{l2ratio} as \(E\rightarrow \infty\). For \(E>0\) and for \(n=2,\cdots,N_E\), consider an operator \(-\Delta_{\mathbf{x}_0,\Omega, \alpha_n}\) and its \(n\)-th lowest eigenvalue \(z'_{n,\alpha_n}\) where \(\mathbf{x}_0=(x_0,\frac{b}{2})\), \(\Omega=[0,E^{\frac{1}{2}}]\times [0,E^{-\frac{1}{2}}]\) and \(\alpha_n\) is given by \eqref{alphan}. We choose \(x_0=0.3 \frac{a}{\pi},~0.7 \frac{a}{\pi},~1.1 \frac{a}{\pi},~1.5 \frac{a}{\pi}\) to observe how the location of a point scatterer affects localization of eigenfunctions. (See Fig.~\ref{fig:degpert}.)  Define \(\epsilon_{n,E}\) as 
\[\epsilon_{n,E}=\begin{dcases}\frac{\|\psi_{z_{n,\alpha_n}'}^{(2)}\|_{L^2(\Omega\setminus\Omega_1)}}{\|\psi_{z_{n,\alpha_n}'}^{(2)}\|_{L^2(\Omega)}}, &\mbox{ if } z_{n,\infty}\in S_1
\\
\frac{\|\psi_{z_{n,\alpha_n}'}^{(2)}\|_{L^2(\Omega_1)}}{\|\psi_{z_{n,\alpha_n}'}^{(2)}\|_{L^2(\Omega)}}, &\mbox{ if } z_{n,\infty}\in S_2
\end{dcases}\]
and define the \textit{rate of localization} as a number \(k\) such that \(\epsilon_{n,E}=O(E^k)\) as \(E\rightarrow\infty\).  Note that \(k<0\) for some \(n \ge 2\) implies that the \(L^2\)-norm fraction of the localized eigenfunction \(\psi_{z'_n,\alpha_n}^{(2)}\) on the unlocalized region (e.g., the right part of Fig.~\ref{mode3}) decays as \(E\) increases. In addition, \(|k|\) measures how fast the unlocalized fraction of \(\psi_{z'_n,\alpha_n}^{(2)}\) diminishes as \(E\rightarrow\infty\) as long as \(k<0\).

\begin{figure}
\centering
\begin{subfigure}{0.49\columnwidth}
\centering
\includegraphics[width=\textwidth]{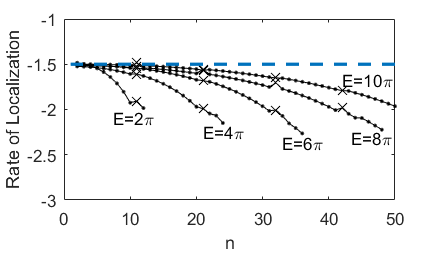} 
\caption{\(x_0=0.3 \frac{a}{\pi}\approx 0.10 \cdot a\)}
\end{subfigure}
\hfill
\begin{subfigure}{0.49\columnwidth}
\centering
\includegraphics[width=\textwidth]{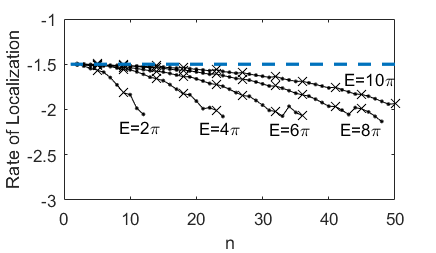} 
\caption{\(x_0=0.7\frac{a}{\pi} \approx 0.22\cdot a\)}
\end{subfigure}

\begin{subfigure}{0.49\columnwidth}
\centering
\includegraphics[width=\textwidth]{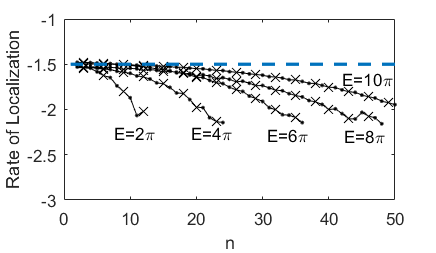} 
\caption{\(x_0=1.1\frac{a}{\pi}\approx 0.35\cdot a\)}
\end{subfigure}
\hfill
\begin{subfigure}{0.49\columnwidth}
\centering
\includegraphics[width=\textwidth]{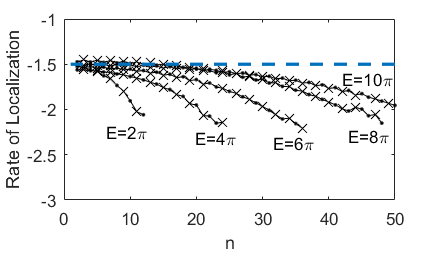} 
\caption{\(x_0=1.5\frac{a}{\pi} \approx 0.48\cdot a\)}
\label{degloc4}
\end{subfigure}

\caption{(\v{S}eba billiard) Rate of localization of eigenfunctions \(\psi_{z'_{n,\alpha_n}}^{(2)},~(2\le n\le \sqrt{3E^2+1})\) approximated with data near \(E=2\pi,4\pi,6\pi,8\pi,10\pi \) with various \(x_0\).  Eigenfunctions localizing in \(\Omega_1\) and \(\Omega\setminus\Omega_1\) were marked as \(\times\) and \(\bullet\), respectively. The dashed lines correspond to the theoretical bound \(k=-1.5\) as \(E\rightarrow\infty\) given by Theorem~\ref{mainthm}} 
\label{fig:degpert}
\end{figure}

The numerical results indicate that \(k \le -\frac{3}{2}\) for any \(n\ge 2\) as predicted by Theorem~\ref{mainthm}.  Fig.~\ref{fig:degpert} shows the rate of localization of \(\psi_{z'_{n,\alpha_n}}^{(2)}\) as a function of \(n\) with various \(x_0\). Data points for eigenfunctions localizing in \(\Omega_1\) and \(\Omega\setminus\Omega_1\) were marked as \(\times\) and \(\bullet\), respectively. According to \eqref{l2rationfixed}, all points in Fig.~\ref{fig:degpert} should lie on or below the horizontal line of \(-\frac{3}{2}\) no matter where the localization occurs in the rectangle. In particular, one can numerically observe that the rate of localization converges to the theoretical bound marked as the dashed line as \(E\rightarrow \infty\)

Note that keeping the irrational ratio between \(x_0\) and \(a\) as assumed in Section~\ref{sec:1d} is crucial to localize all eigenfunctions as \(E\rightarrow \infty\). For example, placing a point scatterer at the midpoint \(x_0=\frac{a}{2}\) does not induce localization of eigenfunctions since \(\Omega_1\) and \(\Omega\setminus\Omega_1\)  are symmetric to each other. However, one can observe that in Fig.~\ref{degloc4} all eigenfunctions localize within a controlled error better than \(O(E^{-\frac{3}{2}})\) although \(x_0=1.5 \frac{a}{\pi}\approx 0.48\cdot a\)  is close to the midpoint of the rectangle. This implies that even a slight change of \(x_0\) can induce a drastic change of eigenfunctions in terms of the localization ratio.

\section{Proofs}\label{proofs}
\begin{proof}[Proof of Lemma~\ref{lemma}]
We rewrite \eqref{psiz2} using \eqref{psiz1fourier} as
\begin{align*}\psi_z^{(2)}(x,y)&= \sum_{n_1=1}^\infty \sum_{n_2=1}^\infty \frac{\overline{\phi_{n_1,n_2}(x_0,y_0) }}{\frac{\pi^2 n_1^2}{E}+\nu_{n_2}E -z}\phi_{n_1,n_2} (x,y)= \psi_{\mathrm{low}}(x,y)+ \psi_{\mathrm{high}}(x,y)
\end{align*}
where
\[\psi_{\mathrm{low}}(x,y)=\sum_{n_1=1}^\infty \sum_{n_2=1}^{\tilde{n}-1} \frac{\overline{\phi_{n_1,n_2}(x_0,y_0) }}{\frac{\pi^2 n_1^2}{E}+\nu_{n_2}E -z}\phi_{n_1,n_2} (x,y),\]
\[\psi_{\mathrm{high}}(x,y)=\sum_{n_1= 1}^\infty \sum_{n_2=\tilde{n}}^\infty \frac{\overline{\phi_{n_1,n_2}(x_0,y_0) }}{\frac{\pi^2 n_1^2}{E}+\nu_{n_2}E -z}\phi_{n_1,n_2} (x,y).\]
By \eqref{phin}, \(\phi_{n_1,n_2}\) can be rewritten as
\[\phi_{n_1,n_2}(x,y)=\sin\left(\frac{n_1 \pi x}{\sqrt{E}}\right) g_{n_2}(\sqrt{E}y).\]
In addition, by \eqref{psiz1fourier},
\begin{align*}\psi_{\mathrm{low}}(x,y)&=\sum_{n_1=1}^\infty \sum_{n_2=1}^{\tilde{n}-1} \frac{\overline{\phi_{n_1,n_2}(x_0,y_0) }}{\frac{\pi^2 n_1^2}{E}+\nu_{n_2}E -z}\phi_{n_1,n_2} (x,y)\\&=\sum_{n_2=1}^{\tilde{n}-1}\overline{g_{n_2}\left(\sqrt{E}y_0\right)}\frac{1}{M}\psi_{z-\nu_{n_2} E}^{(1)}(x) g_{n_2}\left(\sqrt{E}y\right)\\
&= \frac{1}{M}\psi_{z-\nu_{1} E}^{(1)}(x) \left( \sum_{n_2=1}^{\tilde{n}-1}\overline{g_{n_2}\left(\sqrt{E}y_0\right)} g_{n_2}\left(\sqrt{E}y\right)\right) .\end{align*}

 Note that \(\sqrt{E}y_0 = \frac{y_0}{b}\) is a fixed quantity for all \(E>0\) and  \(g_1,g_2\cdots\) given in \eqref{phin} are orthogonal.  By applying \eqref{psi1norm} to \(\psi_{z-\nu_1 E}^{(1)}\),  we obtain \(\|\psi_{\mathrm{low}}\|_{L^2(\Omega)} = C E\) for some constant \(C\ge 0\) since  \(\|\psi_{z-\nu_1 E}^{(1)}\|_{L^2([0,a])}=1\) and 
\begin{equation}\label{psilow}\begin{aligned}\|\psi_{\mathrm{low}} \|_{L^2(\Omega)}^2 &=  \frac{1}{M^2}\|\psi_{z-\nu_1 E}^{(1)}\|_{L^2([0,a])}^2\sum_{n_2=1}^{\tilde{n}-1}\left|g_{n_2}(\sqrt{E}y_0)\right|^2 \left\|g_{n_2}(\sqrt{E}~\bullet)\right\|_{L^2([0,b])}^2\\
&= cE^{\frac{5}{2}} \sum_{n_2=1}^{\tilde{n}-1} \left|g_{n_2}(\sqrt{E}y_0)\right|^2 \left(\|g_{n_2}\|_{L^2([0,1])}^2E^{-\frac{1}{2}}\right)  \quad \text{for some } c
\\&=C^2E^2 \quad \text{for some } C\ge 0.
\end{aligned}\end{equation}

On the other hand, by \eqref{weyl} and \(z <  \lambda_{1,2}\), there exists \(C>0\) such that
\begin{align*}&\|\psi_{\mathrm{high}}\|_{L^2(\Omega)}^2 = \sum_{n_1=1}^\infty \sum_{n_2=\tilde{n}}^\infty \left(\frac{|\phi_{n_1,n_2}(x_0,y_0) |}{\frac{\pi^2 n_1^2}{E}+\nu_{n_2}E -z}\right)^2 \|\phi_{n_1,n_2}\|_{L^2(\Omega)}^2\\&\le C\sum_{n_1=1}^\infty \sum_{n_2=\tilde{n}}^\infty \left(\frac{|\phi_{n_1,n_2}(x_0,y_0) |}{\frac{\pi^2 n_1^2}{E}+n_2^2 C_0 E -z}\right)^2 < C \int_{\frac{\pi}{\sqrt{E}}}^\infty \int_{\tilde{n}\sqrt{C_0 E}}^\infty \left(\left(\xi^2+\eta^2\right)-z\right)^{-2} d\eta d\xi \\&< C \int_0^{\frac{\pi}{2}}\int_{\sqrt{\frac{\pi^2}{E}+\tilde{n}^2 C_0 E }}^\infty \frac{r}{\left( r^2-z\right)^2} dr d\theta = C \frac{1}{\frac{\pi^2}{E}+\tilde{n}^2C_0 E-z}\end{align*}
where \(\tilde{n}, C_0\) are given by \eqref{ntilde} and \eqref{C0}. Therefore,
\begin{equation}\label{psihigh}\|\psi_{\mathrm{high}}\|_{L^2(\Omega)} < C \frac{1}{\sqrt{\frac{\pi^2}{E}+\tilde{n}^2C_0 E-z}}.\end{equation}
By \eqref{psilow}, \eqref{psihigh}, we obtain
\begin{align*}\frac{\|\psi_z^{(2)}\|_{L^2(\Omega_1)}}{\|\psi_z^{(2)}\|_{L^2(\Omega)}}&\le \frac{\|\psi_{\mathrm{low}}\|_{L^2(\Omega_1)} +\|\psi_{\mathrm{high}}\|_{L^2(\Omega_1)}}{\|\psi_{\mathrm{low}}\|_{L^2(\Omega)}-\|\psi_{\mathrm{high}}\|_{L^2(\Omega)}} \\&\le
\frac{\|\psi_{\mathrm{low}}\|_{L^2(\Omega_1)}}{\|\psi_{\mathrm{low}}\|_{L^2(\Omega)}} + C\left(\frac{\|\psi_{\mathrm{high}}\|_{L^2(\Omega_1)}}{\|\psi_{\mathrm{low}}\|_{L^2(\Omega)}}\right)\\&\le
\frac{\|\psi_{\mathrm{low}}\|_{L^2(\Omega_1)}}{\|\psi_{\mathrm{low}}\|_{L^2(\Omega)}} + C\left(\frac{1}{E \sqrt{\frac{\pi^2}{E}+\tilde{n}^2 C_0 E -z}}\right)\\&=\frac{\|\psi_{z-\nu_1 E}^{(1)}\|_{L^2([0,x_0])}}{\|\psi_{z-\nu_1 E}^{(1)}\|_{L^2([0,a])}}+ C\left(\frac{1}{E \sqrt{\frac{\pi^2}{E}+\tilde{n}^2 C_0 E -z}}\right)\end{align*}
and similarly,
\begin{align*}\frac{\|\psi_z^{(2)}\|_{L^2(\Omega_1)}}{\|\psi_z^{(2)}\|_{L^2(\Omega)}}&\ge \frac{\|\psi_{z-\nu_1E}^{(1)}\|_{L^2([0,x_0])}}{\|\psi_{z-\nu_1 E}^{(1)}\|_{L^2([0,a])}}- C\left(\frac{1}{E \sqrt{\frac{\pi^2}{E}+\tilde{n}^2 C_0 E -z}}\right)\end{align*}
for some \(C>0\) as \(E\rightarrow \infty\). This concludes the proof of \eqref{eqn:lemma1}. Similarly, we can prove \eqref{eqn:lemma2} by switching \(\Omega_1\) and \(\Omega\setminus\Omega_1\).
\end{proof}

\begin{proof}[Proof of Theorem~\ref{mainthm}]
Note that \(a=\sqrt{E}\) and \(b=\frac{1}{\sqrt{E}}\) since \(\mathrm{area}(\Omega)=1\). Then we have \(\lambda_n = \lambda_{n,1} \) for all \(n \le N_E\). In addition, \eqref{interlace} implies that for \(n=\tilde{n},\cdots, N_E\) and for any \(\alpha\in \mathbb{R}\), \(z_{n,\alpha}'\) satisfies
\[\lambda_{1,1}= \lambda_{1} \le z'_{n,\alpha} \le \lambda_{N_E} =\lambda_{N_E,1}<\lambda_{1,2}\] 
so we can apply Lemma~\ref{lemma} to \(z'_{n,\alpha} \in \sigma(-\Delta_\alpha)\setminus\sigma(-\Delta)\). Furthermore, for each \(n\), we may choose a specific \(\alpha=\alpha_n\) so that the localization of \(\psi_{z'_{n,\alpha_n}}^{(1)}\) is maximized. More preciesly, consider the 1-dimensional model discussed in Section~\ref{sec:1d} with \(a=\sqrt{E}\). By \eqref{1dinterlace}, we have
\[\lambda_{1,1} < z_{1,\infty} + \nu_1 E< \lambda_{2,1} < z_{2,\infty}+ \nu_1 E < \lambda_{3,1} <z_{3,\infty}+ \nu_1 E < \cdots \]
Hence, for each \(n=\tilde n,\cdots, N_E\), there exists a unique \(\alpha_n\in \mathbb{R}\) given by \eqref{alphan} such that 
\[z'_{n,\alpha_n}= z_{n-1,\infty}+ \nu_1 E \in (S_1\cup S_2) + \nu_1 E\] and by \eqref{1dloc},
\[\begin{dcases}\frac{\|\psi_{z'_{n,\alpha_n}-\nu_1E}^{(1)}\|_{L^2([x_0,a])}}{\|\psi_{z'_{n,\alpha_n}-\nu_1E}^{(1)}\|_{L^2([0,a])}} =0 &\mbox{ if } z'_{n,\alpha_n}-\nu_1E \in S_1\\
\frac{\|\psi_{z'_{n,\alpha_n}-\nu_1E}^{(1)}\|_{L^2([0,x_0])}}{\|\psi_{z'_{n,\alpha_n}-\nu_1E}^{(1)}\|_{L^2([0,a])}} =0 &\mbox{ if } z'_{n,\alpha_n}-\nu_1E \in S_2\end{dcases}\]

Note that 
\(\lambda_{n-1}<z'_{n,\alpha_n} < \lambda_n, \quad \tilde{n}\le n \le N_E\)
 and \[\lambda_{n-1} = \frac{\pi^2 (n-1)^2}{E}+\nu_1 E,\quad \lambda_n = \frac{\pi^2 n^2}{E}+\nu_1 E.\]
Therefore, we can apply Lemma~\ref{lemma} as follows:  If \(z'_{n,\alpha_n}-\nu_1E \in S_2\), as \(E\rightarrow \infty\),
\begin{equation*}
\begin{aligned}\frac{\|\psi_{z'_{n,\alpha_n}}^{(2)}\|_{L^2(\Omega_1)}}{\|\psi_{z'_{n,\alpha_n}}^{(2)}\|_{L^2(\Omega)}} &= 
\frac{\|\psi_{{z'_{n,\alpha_n}}-\nu_1 E}^{(1)}\|_{L^2([0,x_0])}}{\|\psi_{{z'_{n,\alpha_n}}-\nu_1 E}^{(1)}\|_{L^2([0,a])}} + O\left(\frac{1}{E \sqrt{\frac{\pi^2}{E}+\tilde{n}^2 C_0 E -z'_{n,\alpha_n}}}\right)\\&=0 + O\left(\frac{1}{\sqrt{E^3(\tilde{n}^2 C_0-\nu_1)-E \pi^2(n^2-1)}}\right)
\end{aligned}\end{equation*}

 If \(z'_{n,\alpha_n}-\nu_1E \in S_1\), as \(E\rightarrow \infty\),
\begin{equation*}
\begin{aligned}\frac{\|\psi_{z'_{n,\alpha_n}}^{(2)}\|_{L^2(\Omega\setminus\Omega_1)}}{\|\psi_{z'_{n,\alpha_n}}^{(2)}\|_{L^2(\Omega)}} &= 
\frac{\|\psi_{{z'_{n,\alpha_n}}-\nu_1 E}^{(1)}\|_{L^2([x_0,a])}}{\|\psi_{{z'_{n,\alpha_n}}-\nu_1 E}^{(1)}\|_{L^2([0,a])}} + O\left(\frac{1}{E \sqrt{\frac{\pi^2}{E}+\tilde{n}^2 C_0 E -z'_{n,\alpha_n}}}\right)\\&=0 + O\left(\frac{1}{\sqrt{E^3(\tilde{n}^2 C_0-\nu_1)-E \pi^2(n^2-1)}}\right)
\end{aligned}\end{equation*}

\end{proof}

\FloatBarrier

\section*{Acknowledgements}
\thispagestyle{empty}
The author is greatly indebted to Maciej Zworski for suggesting the topic as well as providing guidance throughout the research. The author also thanks Gregory Berkolaiko and the anonymous reviewers for their valuable suggestions that led us to revise the equation~\eqref{psi1} and improve the estimate in the main theorem. 
This research was supported by the Samsung Scholarship.

\nocite{*}
\bibliographystyle{vancouver}
\bibliography{seba_proof_bib}
\end{document}